\documentclass[submission,copyright,creativecommons]{eptcs}
\providecommand{\event}{GandALF 2015} 
\usepackage{breakurl}             
\usepackage[utf8]{inputenc}
\usepackage{amsmath, amssymb, amstext, amsthm, graphicx, enumerate, mathrsfs}
\usepackage{color}
\usepackage{nicefrac} 
\usepackage{stmaryrd} 
\usepackage{graphicx}
\usepackage{xspace}
\usepackage[all,arc,curve,color,frame]{xy}
\usepackage{amsmath, amssymb, amsthm, dsfont}
\usepackage{hyperref}
\usepackage{exscale, relsize} 
\usepackage{amsmath}

\title{The expressive power of modal logic with inclusion atoms}
\author{Lauri Hella
\institute{University of Tampere\\ Tampere, Finland}
\email{lauri.hella@uta.fi}
\and
  Johanna Stumpf
\institute{TU Darmstadt\\
Darmstadt, Germany}
\email{johanna.stumpf@stud.tu-darmstadt.de}
}
\def\titlerunning{The expressive power of modal logic with inclusion atoms}
\def\authorrunning{Lauri Hella, Johanna Stumpf}

\theoremstyle{plain}
\newtheorem{Satz}{Satz}[section]  
	
\newtheorem{Lemma}[Satz]{Lemma}		

\newtheorem{Corollary}[Satz]{Corollary}			

\newtheorem{Theo}[Satz]{Theorem}	
\newtheorem{theorem}[Satz]{Theorem}		

\newtheorem{proposition}[Satz]{Proposition}	
\newtheorem{Proposition}[Satz]{Proposition}	
\theoremstyle{definition}		

\newtheorem{Definition}[Satz]{Definition}		
\newtheorem{example}[Satz]{Example}		
\theoremstyle{remark}			
\newtheorem{Remark}[Satz]{Remark}
\newtheorem{remark}[Satz]{Remark}	

\DeclareMathOperator{\nedis}{\mathsmaller{\overset{{\mathsmaller\bowtie}}{\vee}}}
\DeclareMathOperator{\bnedis}{\raisebox{-.2em}{$\bigvee$}\hspace{-.85em}\raisebox{.65em}{$\bowtie$}}

\DeclareMathOperator{\dep}{=\hspace{-.08em}}

\DeclareMathOperator{\md}{\text{\textnormal{md}}}
\DeclareMathOperator{\nmodels}{\not \models}
\DeclareMathOperator{\neop}{\triangledown \hspace{-0.2em} }
\DeclareMathOperator{\occv}{\text{\textnormal{occ}}_{\triangledown}}

\newcommand{\newLogic}[1]{{\ensuremath{\mbox{{\usefont{OMS}{cmsy}{m}{n}#1}}}}\xspace}
\newcommand{\MDL}{\newLogic{MDL}}
\newcommand{\ML}{\newLogic{ML}}
\newcommand{\EMDL}{\newLogic{EMDL}}
\newcommand{\MINC}{\newLogic{MINC}}

\newcommand{\cL}{\mathcal{L}}

\newcommand{\cP}{\mathcal{P}}
\newcommand{\cK}{\mathcal{K}}
\newcommand{\cT}{\mathcal{T}}

\begin{document}
\maketitle

\begin{abstract}
Modal inclusion logic is the extension of basic modal logic with inclusion atoms, and its semantics is defined on Kripke models with teams. A team of a Kripke model is just a subset of its domain. 
In this paper we give a complete characterisation for the expressive power of modal inclusion logic: a class of Kripke models with teams is definable in modal inclusion logic if and only if it is closed under $k$-bisimulation for some integer $k$, it is closed under unions, and it has the empty team property. We also prove that the same expressive power can be obtained by adding a single unary nonemptiness operator to modal logic.
Furthermore, we establish an exponential lower bound for the size of the translation from modal inclusion logic to modal logic with the nonemptiness operator. 

\end{abstract}

\section{Introduction}
%


Modal inclusion logic, \MINC, is the extension of basic modal logic \ML by \emph{inclusion atoms} of the form 
$$
	\theta:=\varphi_1\ldots\varphi_n\subseteq \psi_1\ldots\psi_n,
$$	
where $\varphi_1\ldots\varphi_n,\psi_1\ldots\psi_n$ are \ML-formulas. The intended meaning of the atom $\theta$ is that, in a given model, any combination of truth values that the tuple $\varphi_1\ldots\varphi_n$ gets, is also realized by the tuple $\psi_1\ldots\psi_n$. 

Note that this idea becomes trivial if the notion of model is as in regular Kripke semantics:  the atom $\theta$ is true in an element $w$ of a Kripke model $K$ just in case $K,w\models\varphi_i\leftrightarrow\psi_i$ for each $i$. In order to get a more meaningful semantics we need to consider sets of elements of Kripke models instead of a single element.  
Thus, the semantics of \MINC is defined on pairs $(K,T)$, where $K$ is a Kripke model and $T$ is a \emph{team} of $K$, i.e., $T$ is a subset of the domain of $K$. The inclusion atom $\theta$ is defined to be true in a team $T$ of a Kripke model $K$ if and only if
for every $w\in T$ there is $v\in T$ such that, for all $1\le i\le n$,
$$
	K,w\models\varphi_i
	\iff K,v\models\psi_i.
$$

\emph{Team semantics} was originally introduced by Hodges \cite{hodges}, who used sets of assignments of first-order variables to define a compositional semantics for the independence-friendly logic of Hintikka and Sandu \cite{if}.
Later Väänänen observed that a dependence between variables can be regarded as an atomic property of teams. In \cite{vbuch}, he introduced dependence logic which is obtained by adding such \emph{dependence atoms} $\dep(x_1, \ldots , x_{n-1}, x_n)$ to first-order logic. The intuitive meaning of this atom is that the value of the variable $x_n$ is functionally determined by the values of~$x_1, \ldots, x_{n-1}$.
Dependence logic was defined with the intention to describe dependences occurring in various scientific disciplines, such as physics, statistics, and even social choice theory.


A team can also be seen as a relational database. In database theory a great variety of constraints over relations have been studied. Galliani \cite{galliani} imported inclusion dependencies of database theory into the team semantics setting by defining inclusion 
atoms of the form~$x_1 \ldots x_n \subseteq y_1 \ldots y_n$. This atom is satisfied in a team $T$ if for every assignment $s\in T$ there is an assignment $t\in T$ such that $(s(x_1),\ldots,s(x_n))=(t(y_1),\ldots,t(y_n))$.
Thus, modal inclusion atoms are the natural variant of the first-order inclusion atoms obtained by replacing the values of first-order variables under assignments with the truth values of formulas in elements of Kripke models.

Team semantics was introduced in the modal context for the first time by Väänänen, who in  \cite{vaananen} defined \emph{modal dependence logic}, \MDL. 
Initially, the research on modal dependence logic was mainly concerned with complexity questions. Sevenster showed  in \cite{sevenster} that the satisfiability problem for \MDL is complete for NEXPTIME.
In \cite{compl}, Lohmann and Vollmer proved complexity results for syntactic fragments of \MDL. Ebbing and Lohmann  \cite{modelcheck} showed that the model checking problem for \MDL is NP-complete. 
Also, in recent years research on complexity questions has been active; see e.g. \cite{mil}, \cite{gandalf} and \cite{yang}. The complexity of the satisfiability problem for modal inclusion logic is studied in \cite{HKMV15}.



The expressive power of \MDL has been studied more systematically only recently. In \cite{emdl}, Ebbing et al.~observed that purely propositional dependence atoms $\dep(p_1,\ldots,p_n,q)$ are too weak to express temporal dependences. To resolve this issue, the authors introduced \emph{extended modal dependence logic}, \EMDL, which allows arbitrary modal formulas in place of the proposition symbols $p_1,\ldots,p_n,q$ in dependence atoms. Furthermore, it was shown in \cite{emdl} that, in terms of expressive power, \EMDL is a sublogic of $\ML(\varovee)$, the extension of modal logic with intuitionistic disjunction $\varovee$.

In \cite{epmdl}, Hella  et~al.~proved that the converse is also true, and hence \EMDL and $\ML(\varovee)$ have equal expressive power. Moreover, they gave a complete characterisation for the expressive power of these two logics: a class of pairs $(K,T)$, where $K$ is a Kripke model and $T$ is a team of $K$, is definable in \EMDL if and only if it is downwards closed, and closed under team $k$-bisimulation for some $k$. Here \emph{downwards closure} is the following property: 
\begin{itemize}
\item for all formulas $\varphi$, if $K,T\models \varphi$ and $S\subseteq T$, then $K,S\models \varphi$. 
\end{itemize}
\emph{Team $k$-bisimulation} is defined by a straighforward lifting of the usual $k$-bisimulation relation to the setting of team semantics: 
\begin{itemize}
\item $T$ and $T'$ are team $k$-bisimilar if and only if every element of $T$ is $k$-bisimilar with some element of $T'$, and vice versa.
\end{itemize}
In addition to these results, \cite{epmdl} proved an exponential lower bound for the size of translation from \EMDL to  $\ML(\varovee)$ by showing that any formula $\ML(\varovee)$ which defines the dependence atom $\dep(p_1,\ldots, p_n,q)$ contains at least $ 2^{n}$ occurrences of the intuitionistic disjunction $\varovee$.  

In this paper we will analyse the expressive power of modal inclusion logic using a similar approach as was used in \cite{epmdl} for extended modal dependence logic. 
As explained above, \EMDL is characterised by two closure properties: downwards closure, and closure under team $k$-bisimulation (for some $k$). 
We first show that closure under $k$-bisimulation holds for all \MINC-definable classes, as well.  However, it is easy to see that downwards closure fails in the case of \MINC. On the other hand, as pointed out by Galliani~\cite{galliani}, first-order inclusion logic is closed under unions and has the empty team property. We show that these properties also hold for \MINC. \emph{Closure under unions} is the following property:
\begin{itemize}
\item for all formulas $\varphi$, if $K,T_i\models \varphi$ for all $i\in I$, then $K,\bigcup_{i\in I} T_i\models \varphi$. 
\end{itemize}
The \emph{empty team property} is the requirement that $K,\emptyset\models\varphi$ for all Kripke models $K$ and formulas $\varphi$. 

Thus, we have identified three natural closure properties of \MINC. In our first main result we prove that \MINC is expressively complete with respect to these closure properties: any class of pairs $(K,T)$ is closed under team $k$-bisimulation for some $k$, closed under unions,  and has the empty team property, if and only if it is definable in \MINC. 




As our second main result, we prove that \MINC is equivalent with an extension $\ML(\triangledown)$ of modal logic by a new unary operator. The semantics of this \emph{nonemptiness operator} $\triangledown$ is given by the clause:
$$
	K,T\models\triangledown\varphi\;\iff\; T=\emptyset\;\text{ or }\;K,S\models\varphi
	\text{ for some nonempty $S\subseteq T$}.
$$
We show that  closure under $k$-bisimulation, closure under unions, and the empty team property hold for $\ML(\triangledown)$. Furthermore, we show that inclusion atoms can be replaced by the use of the nonemptiness operator $\triangledown$ in the proof of definability of classes that have all the three closure properties. Consequently, the expressive power of $\ML(\triangledown)$ satisfies the same characterisation as that of \MINC.


Since \MINC and $\ML(\triangledown)$ have equal expressive power, every formula of the former can be translated to an equivalent formula of the latter (and vice versa). In our third main result, we prove an exponential lower bound for this translation. More precisely, we show that any formula of $\ML(\triangledown)$ that defines the inclusion atom $p_1\ldots p_n \subseteq q_1\ldots q_n$ contains at least $ 2^{n}$ occurrences of the operator $\triangledown$. 

Altogether, our results show that the relationship between the logics \MINC and $\ML(\triangledown)$ is completely analogous to that between \EMDL and $\ML(\varovee)$. However, the method of proof in the exponential lower bound result is quite different in the two cases: in the case of \EMDL and $\ML(\varovee)$, the proof in \cite{epmdl} makes use of a semantic invariant $\mathrm{Dim}(\varphi)$ (upper dimension of $\varphi$) of formulas, while in the case of \MINC and $\ML(\triangledown)$ the proof is obtained by analysing the semantic games of $\ML(\triangledown)$-formulas.




\section{Basic definitions and preliminaries}

In this section, we define the syntax and team semantics of modal logic.  Then we consider
Hintikka-formulas and $k$-bisimulation, which are important tools in analysing the expressive power of modal inclusion logic. Finally, we take a look on the basic closure properties of modal logic with team semantics.

\subsection{Syntax and team semantics}
Basic modal logic is defined as propositional logic with additional unary modal operators. These modal operators are~$\Diamond$, called diamond, and~$\Box $, called box. Box is the dual operator of diamond and its semantics is defined to correspond to~$\neg \Diamond \neg \varphi$. Diamond can be interpreted in such a way that~$\Diamond \varphi$ means `it is possibly the case that~$\varphi$'. Then~$\Box \varphi$ means `it is not possible that not ~$\varphi$' and therefore `necessarily~$\varphi$'. 

It will be useful to assume that all formulas in modal logic are in negation normal form. This means that the negation operator is only applied to proposition symbols. Thus we follow \cite{epmdl} by defining the syntax of basic modal logic \ML as follows:

\begin{Definition}
Let $\Phi$ be a set of proposition symbols. The set of formulas of ~$\ML(\Phi)$ is generated by the following grammar
\[ \varphi := p \mid \neg p \mid (\varphi \land \varphi ) \mid ( \varphi \lor \varphi ) \mid \Diamond \varphi \mid \Box \varphi, \]
where $p \in \Phi$.
\end{Definition}

A \emph{Kripke model} $K$ is a triple $(W, R, V)$, where $W$ is a set, and $R$ is a binary relation on $W$. Elements of $W$ are called nodes or states, and $R \subseteq W \times W$ is known as the accessibility relation. The third component $V: \Phi \rightarrow \cP(W)$ is called the valuation.

%

\begin{Definition} Let $K=(W,R,V)$ be a Kripke model.
\begin{enumerate}[i)]
\item Any subset $T$ of $W$ is called a \emph{team} of $K$.
\item For any $T \subseteq W$ we denote the image of $T$ as $R[T] = \{ v \in W \mid \exists w \in T : wRv\}$ and the preimage as $R^{-1}[T] = \{ w \in W \mid \exists v \in T : wRv\}$. 
\item For teams $T, S \subseteq W$ we write $T[R]S$ if $S \subseteq R[T]$ and $T \subseteq R^{-1}[S]$. 
\end{enumerate}
Thus, $T[R]S$ holds if and only if for every $v \in S$ there exist a $w \in T$ such that $wRv$ and for every $w \in T$ there exist a $v \in S$ such that $wRv$. 
In the case of a singleton team $\{w\}$ we write $R[w]$ and $R^{-1}[w]$ instead of 
$R[\{w\}]$ and $R^{-1}[\{w\}]$.
\end{Definition}

Instead of defining the satisfaction relation with respect to single nodes of the Kripke model, as in regular Kripke semantics (see e.g. \cite{blackburn}), here it is defined with respect to teams. 
\begin{Definition}\label{lax} The team semantics for \ML is due to \cite{vaananen} and defined as follows:
\begin{align*}
K, T &\models p  &&\iff && T \subseteq V(p) \\
K, T &\models \neg p  &&\iff && T \cap V(p) = \emptyset \\
K, T &\models \varphi \land \psi  &&\iff && K,T\models \varphi \text{ and } K,T \models \psi \\
K, T &\models \varphi \lor \psi  &&\iff  && K,T_1 \models \varphi \text{ and } K,T_2 \models \psi  \\ & && &&\text{for some $T_1, T_2$ such that } T_1 \cup T_2 = T.\\
K,T &\models \Diamond \varphi &&\iff  && K, S\models \varphi \text{ for some }S\text{ such that }T[R]S.\\
K,T &\models \Box \varphi &&\iff && K, S \models \varphi \text{ where } S= R[T].
\end{align*}
\end{Definition}

\begin{remark}
The semantics given in Definition \ref{lax} is so-called \emph{lax semantics} for \ML. There is alternative \emph{strict semantics} with different truth conditions for disjunction and diamond, see~\cite{HKMV15}. However, with strict semantics \ML does not have the crucial property of being closed under unions.
\end{remark}

For singleton teams $T=\{w\}$, $K, T \models \varphi$ is equivalent to $K,w \models \varphi$ in the regular Kripke semantics. Therefore, we can write $K,w \models \varphi$ instead of $K, \{w\} \models \varphi$ without ambiguity. 

A team satisfies a formula of \ML if and only of each node in the team satisfies the formula. This is called the \emph{flatness property}. The flatness property does not hold in the extensions of modal logic which we will study in the following sections. 
\begin{proposition} \label{T->v} (\cite{epmdl}, Proposition 2.5.)
Let $K$ be a Kripke model, $T$ a team of $K$, and $\varphi$ an $\ML(\Phi)$-formula. Then
\[K,T \models \varphi \iff K, w \models \varphi \text{ for every } w \in T.\]
\end{proposition}



\subsection{Bisimulation and Hintikka-formulas}
An important concept while dealing with modal logics is bisimulation. Following \cite{blackburn} a bisimulation is a relation between two models in which related states satisfy the same proposition symbols and have matching accessibility possibilities. Here we use the notion of \emph{$k$-bisimulation} between pointed Kripke models $(K, w)$ and $(K,' w')$ where accessibility possibilities for $w$ and $w'$ match up to a certain degree $k$.

\begin{Definition} 
 Let $K$ and $K'$ be Kripke models and let $w$ and $w'$ be states of $K$ and $K'$. The~\emph{$k$-bisimulation relation} between $(K, w)$ and $(K,' w')$, denoted as $ K, w \rightleftarrows_k K', w'$, is defined recursively as follows:
\begin{enumerate}[i)]
\item $ K, w \rightleftarrows_0 K', w'$ if and only if the equivalence $K, w \models p \iff K', w' \models p $ holds for all $p \in \Phi$.
\item $ K, w \rightleftarrows_{k+1} K', w'$ if and only if $ K, w \rightleftarrows_0 K', w'$ and 
\begin{itemize}
\item for every state $v$ of $K$ with $wRv$ there is a state $v'$ of $K'$ with $w'R'v'$ such that $ K, v \rightleftarrows_{k} K', v'$ 
\item for every state $v'$ of  $K'$ with $w'R'v'$ there is a state $v$ of $K$ with $wRv$ such that $ K, v \rightleftarrows_{k} K', v'$.
\end{itemize}
\end{enumerate}
\end{Definition}

We write $K,w \not \leftrightarrows_k K',w'$ if $(K, w)$ and $(K,' w')$ are not $k$-bisimilar. 




\begin{Definition}\label{md}
The modal depth md($\varphi$) of a formula of $\ML(\Phi)$ is defined in in the following way: 
\begin{align*}
\md(p) &=\md(\neg p)=0 \text{ for } p \in \Phi,  \\
\md(\varphi \land \psi)&=\md(\varphi \lor \psi)= \max \{\md (\varphi), \md (\psi)\},\\ 
 \md(\Diamond \varphi) &= \md(\Box \varphi) = \md (\varphi) + 1.
\end{align*}
\end{Definition}

If two pointed $\Phi$-models~$(K, w)$ and $(K,' w')$ agree on all modal formulas of modal depth at most $k$ we call them \emph{$k$-equivalent}.
\begin{Definition}
We say that $(K, w)$ and $(K,' w')$ are \emph{$k$-equivalent}, $ K, w \equiv_k K', w'$, if for every ~$\varphi \in \ML(\Phi)$ with ~$\md(\varphi) \leq k$
\[K, w \models \varphi \iff K', w' \models \varphi .\]
\end{Definition}

We will also make use of the fact that for every pointed $\Phi$-model~$(K, w)$ and every~$k \in \mathbb{N}$ there is a formula that characterises~$(K, w)$ completely up to $k$-equivalence. These \emph{Hintikka-formulas}, which are also called characteristic formulas, are defined as in \cite{otto}.
\begin{Definition}
Assume that $\Phi$ is a finite set of proposition symbols. Let $k \in \mathbb{N}$ and let $(K, w)$ be a pointed $\Phi$-model. The \emph{$k$-th Hintikka-formula} $\chi_{K,w}^k$ of $(K, w)$ is defined recursively as follows:
\begin{itemize}
\item $\chi_{K,w}^0:= \bigwedge \{p \mid p \in \Phi, w \in V(p)\} \land \bigwedge \{ \neg p \mid p \in \Phi, w \notin V(p)\}$
\item $\chi_{K,w}^{k+1}:= \chi_{K,w}^k \land \bigwedge_{v \in R[w]} \Diamond  \chi_{K,v}^k \land \Box \bigvee_{v \in R[w]}  \chi_{K,v}^k.$
\end{itemize}
\end{Definition}

Note that since $\Phi$ is finite, there are only finitely many different Hintikka-formulas $\chi_{K,w}^k$ for each $k\in\mathbb{N}$. 
It is easy to see that $\md  (\chi_{K,w}^k) = k$, and $K, w \models  \chi_{K,w}^k$ for every pointed $\Phi$-model $(K, w)$. Moreover, the Hintikka-formula  $\chi_{K,w}^k$ captures the essence of $k$-bisimulation:

\begin{proposition}\label{hintikka} (\cite{otto}, Theorem 32) 
Let $\Phi$ be a finite set of proposition symbols,~$k \in \mathbb{N}$, and~$(K, w)$ and~$(K,' w')$ pointed~$\Phi$-models. Then the following holds:
\[  K, w \equiv_k K', w' \iff  K, w \rightleftarrows_k K', w' \iff K', w' \models \chi_{K,w}^k. \]
\end{proposition}

The expressive power of basic modal logic can be characterised via bisimulation. A proof of this well-known result of van Benthem can for instance be found in \cite{blackburn}. 

\begin{theorem}
A class~$\mathcal{K}$ of pointed Kripke models~$(K,w)$ is definable by a formula of modal logic if and only if~$\mathcal{K}$ is closed under $k$-bisimulation for some~$k \in \mathbb{N}$.
\end{theorem}

Next we define \emph{team $k$-bisimulation}, the canonical adaption of $k$-bisimulation to team semantics.
A \emph{$\Phi$-model with a team} is a pair $(K,T)$, where $K$ is a Kripke model over $\Phi$ and $T$ is a team of $K$. We denote by $\cK \! \!\cT (\Phi)$ the class of all $\Phi$-models with teams.

\begin{Definition}(\cite{epmdl}, Definition 3.1)
Let $(K,T)$, $(K', T') \in \cK\! \! \cT (\Phi)$ and $k \in \mathbb{N}$. We say that $(K,T)$ and $(K', T')$ are \emph{team $k$-bisimilar} and denote $ K, T [ \rightleftarrows_k ] K', T'$ if the following \emph{domain} and \emph{range totality} conditions hold: 
\begin{enumerate}
\item[$(\mathrm{D}_k)$] for every $w \in T$ there exists some $w' \in T'$ such that $ K, w \rightleftarrows_k K', w'$ 
\item[$(\mathrm{R}_k)$] for every $w' \in T'$ there exists some $w \in T$ such that $ K, w \rightleftarrows_k K', w'$  
\end{enumerate}
\end{Definition}

We write $K, T [ \not \rightleftarrows_k] K',T'$ if $(K, T)$ and $(K,' T')$ are not $k$-bisimilar. 
We say that a class $\mathcal{K}$ is \emph{closed under team $k$-bisimulation} if  $(K,T) \in \mathcal{K}$ and $K, T [ \rightleftarrows_k] K',T'$ imply that $(K', T') \in \mathcal{K}$.

\begin{Lemma}\label{n<k}(\cite{epmdl}, Lemma 3.2)
Let $(K, T), (K',T') \in \cK\!\!\cT(\Phi)$ and $k \in \mathbb{N}$. If $K, T [\rightleftarrows_k] K',T'$, then \mbox{} \mbox{$K, T [\rightleftarrows_n] K',T'$} for all $n \leq k$. 
\end{Lemma}

This is easy to prove by using the definition of $k$-bisimulation. The following lemma is also a straightforward consequence of the definition of team $k$-bisimulation.

\begin{Lemma}\label{teilmengen}(\cite{epmdl}, Lemma 3.3)
Let  $k \in \mathbb{N}$ and assume that $(K, T), (K',T') \in \cK\!\!\cT(\Phi)$ are such that $K, T [\rightleftarrows_{k+1}] K',T'$. Then
\begin{enumerate}[i)]
\item for every $S$ such that $T[R]S$ there is an $S'$ such that $T'[R']S'$ and $K, S [\rightleftarrows_{k}] K', S';$
\item for every $S'$ such that $T'[R']S'$ there is an $S$ such that $T[R]S$ and $K, S [\rightleftarrows_{k}] K', S';$
\item  $K, S [\rightleftarrows_{k}] K', S'$ for $S= R[T]$ and $S'= R'[T'];$
\item for all $T_1, T_2 \subseteq T$ such that $T= T_1 \cup T_2$ there are $T'_1, T'_2 \subseteq T'$ such that  $T'= T'_1 \cup T'_2$, \\ and~$K, T_i [\rightleftarrows_{k+1}] K', T'_i$ for $i \in \{1,2\}$.
\end{enumerate}
\end{Lemma}

\subsection{Closure properties of logics}

Let $\cL$ be a modal logic  with team semantics which is an extension of \ML. Each formula  $\varphi \in \cL (\Phi)$ defines a class of $\Phi$-models with teams
\[ \| \varphi \|:= \{ (K,T) \in \cK \! \!\cT (\Phi) \mid  K, T \models \varphi\},\] containing all models with teams satisfying $\varphi$.
%
Similarly, given a Kripke model $K=(W,R,V)$ over $\Phi$, each formula $\varphi \in \cL (\Phi)$ defines a set of teams of $K$
\[\| \varphi \|^K:= \{T \subseteq W \mid K,T \models \varphi \}.\]
%
We say that $\cL$ is \emph{downwards closed}, if for all $\varphi \in \cL$ it holds that $T \in \| \varphi \|^K$ and $S \subseteq T$ imply $S \in \| \varphi \|^K $. Thus, if a team $T$ satisfies a formula $\varphi$ every subteam $S$ of $T$ satisfies this formula as well. 
Furthermore, $\cL$ is \emph{closed under unions}, if for all $\varphi \in \cL$ it holds that if $T_i$, $i \in I$, is a collection of teams of $K$ such that $K,T_i \models \varphi \text{ for all } i \in I$, then $K, \bigcup_{i \in I} T_i \models \varphi$. 
Finally, we say that $\cL$ has the \emph{empty team property}, if $\emptyset\in \| \varphi \|^K$ for all $\varphi\in\cL$ and all Kripke models $K$.

It is easy to show that basic modal logic has all these closure properties. 
\begin{Proposition} \label{mlunion}  
\ML has the following closure properties: \\ 
(a) (\cite{vaananen}, Lemma 4.2) $\ML$ is downwards closed.\\
(b) $\ML$ is closed under unions.\\
(c) $\ML$ has the empty team property.
\end{Proposition}

The proof of (b) is similar to the proof for (first-order) inclusion logic by Galliani in \cite{galliani}. The proof of (c) is by straightforward induction on formulas of $\ML$. 


\section{Modal logic with inclusion atoms}

In this section we introduce modal logic with inclusion atoms, \MINC, and give a characterisation for its expressive power in terms of $k$-bisimulation, closure under unions and empty team property.


\begin{Definition}
Let $\Phi$ be as set of proposition symbols. The syntax of $\MINC(\Phi)$ is obtained by the following grammar
\[ \varphi := p \mid \neg p \mid (\varphi \land \varphi ) \mid ( \varphi \lor \varphi ) \mid \Diamond \varphi \mid \Box \varphi \mid \varphi_1\ldots\varphi_n\subseteq\psi_1\ldots\psi_n, \] where $p\in\Phi$ and $\varphi_i, \psi_i \in \ML$ for $1 \leq i \leq n$. 
\end{Definition}

Note that by this definition, no nesting of inclusion atoms is allowed.

\begin{Definition} The semantics for \MINC is given by the semantics for $\ML$ and the 
additional clause:
\begin{align*}
K, T \models  \varphi_1\ldots\varphi_k\subseteq\psi_1\ldots\psi_k  \;\iff\; 
 \forall w \in T  \ \exists v \in T: \bigwedge_{i=1}^n (K, w \models \varphi_i \iff K, v \models  \psi_i).
\end{align*}

\end{Definition}

We follow the convention of \cite{epmdl} by calling formulas of the form $\varphi_1\ldots\varphi_n\subseteq\psi_1\ldots\psi_n$ inclusion \emph{atoms}, even though they contain formulas instead of just proposition symbols. If only proposition symbols in inclusion atoms were allowed, the expressive power of the logic would be restricted. Indeed, purely propositional inclusion atoms are not sufficient as the following example shows.
%

\begin{example}
Let $\Phi=\{p\}$ and let $K$ be a Kripke model where $W=\{w, v\}$, $R= \emptyset$ and $V(p)= \{v\}$.
The inclusion atom $\varphi=p\subseteq \neg p$ is not definable in terms of purely propositional inclusion atoms. 
As $\Phi$ consists only of one proposition symbol all propositional inclusion atoms are of the form $p\ldots p\subseteq p\ldots p$, and they are trivially true in all teams of $K$. Thus, they can be replaced by $p\lor\lnot p$. On the other hand consider the team $T= \{w, v\}$. It holds that $K,w \nmodels p$ and $K, v \models p$. Thus, this team satisfies $\varphi$, but none of its singleton subsets $\{w\}, \{v\}$ satisfies $\varphi$. So $\| \varphi\|^K$ is not downwards closed and therefore $\varphi$ is not definable by any formula of modal logic with purely propositional inclusion atoms.
\end{example}

We will need the definition of modal depth in modal logic with inclusion atoms.

\begin{Definition}
The modal depth md($\theta$) of a formula of $\MINC$ is defined by adding the case 
\[\md(\varphi_1\ldots \varphi_n \subseteq \psi_1\ldots\psi_n)=\max\{\md(\varphi_1),\ldots, \md(\varphi_n), \md(\psi_1),\ldots, \md(\psi_n) \} \] to Definition \ref{md}.
\end{Definition}

We prove now that each formula of modal inclusion logic is closed under $k$-bisimulation for some $k$. 

\begin{Proposition}\label{k-bisim-closure}
Let $\Phi$ be a set of proposition symbols and let $\cK \subseteq \cK\! \!\cT (\Phi)$. If $\cK$ is definable in \MINC, then there exists a $k \in \mathbb{N}$ such that $\cK$ is closed under $k$-bisimulation.
\end{Proposition}
\begin{proof}
Assume that $ \varphi \in \MINC$. We prove by induction on $\varphi$ that the class $\| \varphi \|$ is closed under $k$-bisimulation, where $k= \md(\varphi)$.
\begin{enumerate}[-]
\item Let $\varphi=p \in \Phi$ and assume that $K, T \models \varphi$ and $K, T [\rightleftarrows_k]K', T'$ for $k=0$. Then $K,w \models p$ for all $w \in T$, and for each $w' \in T'$ there is $w \in T$ such that $K, w \rightleftarrows_0 K', w'$. Thus, for all $w' \in T'$, $K', w' \models p$, whence $K', T' \models \varphi$.
\item The case $\varphi= \neg p$ is similar to the previous one.
\item Let $\varphi= \psi \lor \theta$, and assume $K, T \models \varphi$ and $K, T [\rightleftarrows_k]K', T'$, where $k = \md (\varphi)= \max\{\md(\psi), \md (\theta)\}$. Then there are $T_1, T_2 \subseteq T$ such that $T = T_1 \cup T_2$, $K, T_1 \models \psi$ and $K, T_2 \models \theta$. By Lemma~\ref{teilmengen} (iv) there are subteams $T'_1, T'_2 \subseteq T'$ such that $T' = T'_1 \cup T'_2$ and $K, T_i [\rightleftarrows_k]K', T'_i$ for $i \in \{1,2\}$, whence $K, T_1 [\rightleftarrows_m]K', T'_1$ and $K, T_2 [\rightleftarrows_n] K', T'_2$, where $m = \md (\psi)$ and $n = \md(\theta)$. By induction hypothesis, $K', T'_1 \models \psi$ and $K,T'_2 \models \theta$. Thus, $K', T' \models \varphi$.
\item Let $\varphi= \psi \land \theta$, and assume $K, T \models \varphi$ and $K, T [\rightleftarrows_k]K', T'$, where $k = \md (\varphi)= \max\{\md(\psi), \md (\theta)\}$. Then $K, T \models \psi$ and $K, T \models \theta$. By Lemma~\ref{n<k} it holds that $K, T [\rightleftarrows_m]K', T'$ and $K, T [\rightleftarrows_n] K', T'$, where $m = \md (\psi)$ and $n = \md(\theta)$. Thus, by induction hypothesis $K', T' \models \psi$ and $K', T' \models \theta$. Hence, $K', T' \models \varphi.$
\item Let  $\varphi= \varphi_1\ldots \varphi_n \subseteq \psi_1\ldots\psi_n$ and assume that $K,T \models \varphi$ which means that for each $ w \in T$  there exist $v \in T$ such that $\bigwedge_{i=1}^n (K, w \models \varphi_i \iff K,v \models  \psi_i)$. Let $K,T [\rightleftarrows_k] K',T'$ where $k=\md(\varphi)$. Then for every $w' \in T'$ there exists $w  \in T$ such that $K, w \leftrightarrows_k K',w'$ and for every $w  \in T$ there exists $w' \in T'$  such that $K, w \leftrightarrows_k K',w'$. By Proposition \ref{hintikka} it follows that for every $w' \in T'$ there exists $w  \in T$ such that $K', w' \models \varphi_i \iff K, w \models \varphi_i$ for each $1 \leq i \leq n$, since $\md(\varphi_i) \leq k$. By assumption for each of these $w \in T$ there exists a $v \in T$ such that $K, w \models \varphi_i \iff K,v \models  \psi_i$.  Again by the definition of $k$-bisimulation for each $v \in T$ there exists a $v' \in T'$ such that $K', v' \models \psi_i \iff K, v \models \psi_i$ for each of the given $\psi_i$. Therefore, for each $w' \in T'$ we find a $v' \in T'$ such that $K', w' \models \varphi_i \iff K,v' \models  \psi_i$ for all $i \in I$ and thus $K', T' \models \varphi$. 

\item Let $\varphi = \Diamond \psi$, and assume that $K, T \models \varphi$ and $K, T [\rightleftarrows_k]K', T'$, where $k = \md (\varphi)= \md(\psi)+1$. Then there is a team $S$ of $K$ such that $T[R]S$ and $K, S \models \psi$. By Lemma~\ref{teilmengen} (i), there is a team $S'$ such that $T'[R']S'$ and $K, S [\rightleftarrows_{k-1}] K', S'$. By induction hypothesis $K',S' \models \psi$, and consequently $K', T' \models \varphi$.
\item Let $\varphi = \Box \psi$, and assume that $K, T \models \varphi$ and $K, T [\rightleftarrows_k]K', T'$, where $k = \md (\varphi)= \md(\psi)+1$. Then $K, R[T] \models \psi$ and by Lemma~\ref{teilmengen} (iii), $K, R[T] [\rightleftarrows_{k-1}] K', R'[T']$. Thus, by induction  \mbox{ hypothesis } \mbox{$K', R[T'] \models \psi$} and consequently $K', T' \models \varphi$. \qedhere

\end{enumerate}

\end{proof}

%
%
%

By Proposition~\ref{mlunion}(b), modal logic is closed under unions. This is proved 
by structural induction on $\varphi\in\MINC$. We will add the case of the inclusion atoms to this proof. 

\begin{Proposition} \label{mincunion}
\MINC is closed under unions.
\end{Proposition}
\begin{proof} The proof is by induction on $\varphi\in\MINC$. The steps for literals, connectives, diamond and box are 
as in the case of $\ML$.
Let $\varphi:= \varphi_1\ldots \varphi_n \subseteq \psi_1\ldots \psi_n$. Assume $K, T_i \models \varphi$ for all $ i \in I$. Then for all $w \in T_i$ there exists $ v \in T_i $ such that $ \bigwedge_{i=1}^n(K, w \models \varphi_i \iff K, v \models \psi_i)$. Therefore for all $w \in \bigcup_{i \in I} T_i$ there exists $v \in \bigcup_{i \in I} T_i $ such that $ \bigwedge_{i=1}^n(K, w \models \varphi_i \iff K, v \models \psi_i)$. Thus, $K, \bigcup_{i \in I} T_i \models \varphi$.
\end{proof}

It is also straightforward to prove by induction that $K,\emptyset\models\varphi$ for every formula $\varphi$ of modal inclusion logic: 

\begin{Proposition}\label{emptyteam}
\MINC has the empty team property. 
\end{Proposition}
\begin{proof}
By Proposition~\ref{mlunion} (c) \ML has the empty team property. The proof is done by structural induction on $\varphi$. We add the case of $\varphi$ being an inclusion atom. 
 
Let $\varphi= \varphi_1\ldots \varphi_n \subseteq \psi_1\ldots \psi_n $. Then $K,T \models \varphi$ if and only if  for all $w \in T$ there exists $v \in T$ such that $\bigwedge_{i=1}^n (K, w \models \varphi_i \iff K, v \models  \psi_i)$. By induction hypothesis $K, \emptyset \models \varphi_i$ and $K, \emptyset \models \psi_i$ for $1 \leq i \leq n$. Thus, $K, \emptyset \models\varphi $. 
\end{proof} 

By Propositions \ref{k-bisim-closure}, \ref{mincunion} and \ref{emptyteam}, every $\MINC$-definable class is closed under team $k$-bisimulation for some $k$, closed under  unions, and has the empty team property. Our aim is to prove that the converse is also true: if a class has these three closure properties, then it is definable in \MINC. We start by showing that the range totality condition $(\mathrm{R}_k)$ in the definition of team $k$-bisimilarity can be described by a formula of \ML.

\begin{Lemma}\label{k-equiv1}
Let $\Phi$ be a finite set of proposition symbols and $k \in \mathbb{N}$. For any pair $(K,T) \in \cK\! \!\cT (\Phi)$ there exists a formula $\eta^k_{K,T} \in \ML (\Phi)$ such that 
\[K',T' \models \eta^k_{K,T } \iff\forall w'\in T'\exists w\in T: K,w\rightleftarrows_k K',w'.\]
\end{Lemma}

\begin{proof}
If $T=\emptyset$, then the right-hand side of the equivalence holds if and only if $T'=\emptyset$. Thus, in this case we can simply set $\eta^k_{K,T } = p\land \lnot p$, where $p\in \Phi$.
Else we let $\eta^k_{K,T}$ to be the formula \[ \bigvee_{w \in T } \chi_{K,w}^k,\] where $\chi_{K,w}^k$ is the $k$-th Hintikka-formula of the pair $(K, w)$. Note that since $\Phi$ is finite, there are only finitely many different Hintikka-formulas $\chi_{K,w}^k$. Thus, the disjunction $\bigvee_{w \in T}$ is essentially finite, whence $\eta^k_{K,T} \in \ML$.

Assume that for every $w' \in T'$ there exists  $w \in T$ such that $ K, w \rightleftarrows_k K', w'$. Thus, by Proposition \ref{hintikka}, for every $w' \in T'$ there exists $w \in T$ such that $K',w' \models \chi_{K,w}^k$. Let $T'_w=\{ w' \in T' \mid K',w' \models \chi_{K,w}^k \}$ for each $w\in T$. Then $T' = \bigcup_{w\in T} T'_w$, and by  Proposition~\ref{T->v}, $K', T'_w \models \chi_{K,w}^k $. It follows that $K',T' \models \bigvee_{w \in T } \chi_{K,w}^k $. 

For the other direction assume that $K',T' \models \eta^k_{K,T}$. Thus, there exist $T'_w \subseteq T'$, $w\in T$, such that $T'= \bigcup_{w \in T} T'_w$ and $K', T'_w \models \chi_{K,w}^k$. It follows now from Proposition~\ref{T->v} and Proposition \ref{hintikka} that for all $w' \in T'$ there exists  $w \in T$ such that $K,w \rightleftarrows_k K',w'$. 
\end{proof}

Next we show that, by using inclusion atoms, Lemma \ref{k-equiv1} can be extended to cover also the domain totality condition $(\mathrm{D}_k)$ of team $k$-bisimilarity (except for the case $T'=\emptyset$). 

\begin{Lemma}\label{k-equiv}
Let $\Phi$ be a finite set of propositon symbols and $k \in \mathbb{N}$. For any pair $(K,T) \in \cK\! \!\cT (\Phi)$ there exists a formula $\psi^k_{K,T} \in \MINC (\Phi)$ such that $K',T' \models \psi^k_{K,T }$ if and only if $K, T [\rightleftarrows_k] K',T'$ or $T'=\emptyset$.
\end{Lemma}
\begin{proof}
If $T$ is the empty team, then $K, T [\rightleftarrows_k] K',T'$ holds if and only if $T'$ is empty, as well. Thus, in this case we can simply set $\eta^k_{K,T } = p\land \lnot p$, where $p\in \Phi$.
Else we let $\psi^k_{K,T}$ to be the formula 
\[\eta_{K,T}^k\land\bigwedge_{u, v \in T} \left( \chi_{K,u}^k  \subseteq \chi_{K,v}^k \right),\] 
where $\eta_{K,T}^k$ is as in Lemma~\ref{k-equiv1}. The conjunction $\bigwedge_{u,v \in T }$ is finite by the same argument as in the proof of Lemma~\ref{k-equiv1},  whence $\psi^k_{K,T} \in \MINC$.

Assume first that $K',T' \models \psi^k_{K,T}$, but $T' \neq \emptyset$. To prove that $K, T [\rightleftarrows_k] K',T'$ it suffices to show that for every $w\in T$ there exists $w'\in T'$ such that $K, w \rightleftarrows_k K',w'$, since the converse holds by Lemma~\ref{k-equiv1}. Thus, let $v\in T$. Choose an arbitrary $u'\in T'$. Since $K',T'\models\eta_{K,T}^k$, there is $u\in T$ such that $K',u'\models\chi_{K,u}^k$. Furthermore, since $K',T'\models \chi_{K,u}^k  \subseteq \chi_{K,v}^k$, there exists $v'\in T'$ such that $K',v'\models\chi_{K,v}^k$. It follows  now from Proposition \ref{hintikka} that $K, v \rightleftarrows_k K',v'$.

For the other direction assume that $K, T [\rightleftarrows_k] K',T'$ or $T'=\emptyset$. If $T'=\emptyset$, then $K',T'\models\psi^k_{K,T}$ by the empty team property, so we may assume that $K, T [\rightleftarrows_k] K',T'$. Then for every $w'\in T'$ there exists $w\in T$
such that $K, w \rightleftarrows_k K',w'$, and hence by Lemma~\ref{k-equiv1}, $K',T'\models\eta_{K,T}^k$. 

We still need to prove that each inclusion atom $\chi_{K,u}^k  \subseteq \chi_{K,v}^k $, $u,v\in T$, is satisfied by the pair $(K',T')$.  Note that we may assume that $\chi_{K,u}^k$ and $\chi_{K,v}^k$ are mutually exclusive, since otherwise $\chi_{K,u}^k=\chi_{K,v}^k$ which means that $\chi_{K,u}^k  \subseteq \chi_{K,v}^k$ is valid. Since $K, T [\rightleftarrows_k] K',T'$, there exist $u',v'\in T'$ such that $K, u \rightleftarrows_k K',u'$ and $K, v \rightleftarrows_k K',v'$, and hence by Proposition~\ref{hintikka}, $K',u'\models \chi_{K,u}^k$ and $K',v'\models \chi_{K,v}^k$. Since $\chi_{K,u}^k$ and $\chi_{K,v}^k$ are mutually exclusive, it follows that $K',u'\nmodels \chi_{K,v}^k$. Thus, for every $w'\in T'$ either $K',w'\models\chi_{K,u}^k\iff K',u'\models\chi_{K,u}^k$, or $K',w'\models\chi_{K,u}^k\iff K',v'\models\chi_{K,u}^k$, and hence $K',T'\models \chi_{K,u}^k  \subseteq \chi_{K,v}^k$.
\end{proof}


We are finally ready to give the promised characterisation for the expressive power of \MINC. 

\begin{Theo} \label{mincchar} Let $\Phi$ be a finite set of proposition symbols and let $\cK \subseteq \cK\!\!\cT(\Phi)$. The class $\cK$ is definable in $\MINC$ if and only if it is closed under unions,  closed under $k$-bisimulation for some $k \in \mathbb{N}$ and has the empty team property. 
\end{Theo}

\begin{proof}
Assume first that $\cK$ is definable in $\MINC$. Then by Proposition \ref{k-bisim-closure}, $\cK$ is closed under $k$-bisimulation for some $k \in \mathbb{N}$. By Proposition \ref{mincunion} it is closed under unions, and by Proposition \ref{emptyteam}, it has the empty team property.

For the other direction assume that  $\cK$ is closed under unions,  closed under $k$-bisimulation for some $k \in \mathbb{N}$ and has the empty team property.
Let $\varphi$ be the formula
\[ \bigvee_ {(K,T) \in \cK} \psi_{K,T}, \]
where  $\psi_{K,T}$ is defined as in the proof of Lemma \ref{k-equiv}. 
We prove that $\varphi$ defines $\cK$.

Assume first that $(K , T ) \in \cK$. By Lemma \ref{k-equiv} it holds that $K,T \models \psi_{K,T}$ and consequently, $K,T  \models \varphi$.

Assume for the other direction that $K' , T'  \models \varphi$. Then there are subsets $T'_{K,T} \subseteq T'$,  such that $K', T'_{K,T} \models \psi_{K,T}$ and $T' = \bigcup_{(K,T) \in \cK} T'_{K,T}$. By Lemma \ref{k-equiv} it holds that either  $K,T [\rightleftarrows_k] K', T'_{K,T}$ or $T'_{K,T}= \emptyset$. As the class $\cK$ is closed under $k$-bisimulation and has the empty team property, it follows that $(K',T'_{K,T}) \in \cK$. By closure under unions it follows that $(K',T')~\in~\cK$. \qedhere
\end{proof}


\section{Modal logic with nonemptiness operator}

In the proof of Lemma~\ref{k-equiv} inclusion atoms are used for expressing that for each $w\in T$, the Hintikka-formula $\chi_{K,w}^k$ is satisfied in at least one node of the team $T'$. In this section we show that the same can be achieved by adding a simple unary nonemptiness operator $\triangledown$ to  \ML. 

\begin{Definition}
Let $\Phi$ be as set of proposition symbols. The syntax of $\ML(\triangledown)(\Phi)$ is obtained by the following grammar
\[ \varphi := p \mid \neg p \mid (\varphi \land \varphi ) \mid ( \varphi \lor \varphi ) \mid \Diamond \varphi \mid \Box \varphi \mid  \neop \varphi. \] 
\end{Definition}

The semantics of the nonemptiness operator $\triangledown$ is defined as follows:
\begin{Definition}
$K, T \models \neop \varphi$ if and only if $T= \emptyset$ or there exists $S \subseteq T$ such that $S \neq \emptyset$ and $K,S \models \varphi$.
\end{Definition}

Next we will need the definition of modal depth in modal logic with nonempty disjunction. 

\begin{Definition}
The modal depth md($\theta$) of a formula of $\ML(\triangledown)$ is defined by adding the case of $ \neop \varphi$ 
\[\md(\neop \varphi)= \md (\varphi) \]  to Definition \ref{md}.
\end{Definition}

Now we can show that $\ML(\triangledown)$ is closed under $k$-bisimulation for some $k \in \mathbb{N}$. 
\begin{Proposition}\label{k-bisim-closure2}
Let $\Phi$ be a set of proposition symbols and let $\cK \subseteq \cK \!\! \cT (\Phi)$. If $\cK$ is definable in $\ML(\triangledown)$, then there exists a $k \in \mathbb{N}$ such that $\cK$ is closed under $k$-bisimulation.
\end{Proposition}

\begin{proof}
Assume that $ \varphi \in \ML(\triangledown)$. In Proposition~\ref{k-bisim-closure} it was shown  by induction on $\theta$ that the class $\| \theta \|$ is closed under $k$-bisimulation for all $\theta \in \MINC$, where $k= \md(\theta)$. Since the cases for literals, conjunction, disjunction as well as diamond and box are the same, will just add the case of $\varphi$ containing the nonemptiness operator. 

Let $\varphi=  \neop \psi$ and let $K,T \models \varphi$. Assume that $T\neq \emptyset$, as otherwise the proof is trivial. Then there exists $S \subseteq T$ such that $S \neq \emptyset$ and  $K, S \models \psi$. Let $K,T [\rightleftarrows_k] K',T'$ where $k=\md(\varphi)= \md(\psi) $. Let  $S'= \{w' \in T' \mid K, w \leftrightarrows_k K', w' \text{ for some } w \in S \}$. Then by the definition of $k$-bisimulation it holds that $S'$ is a nonempty subsets of $T'$. Additionally it holds that $K,S [\rightleftarrows_k] K',S'$. Thus, by induction hypothesis it follows that $K', S' \models \psi$. Therefore, $K', T' \models \varphi$.  
\end{proof}

We can easily show that $\ML(\triangledown)$ is closed under unions.

\begin{Proposition} \label{nedisunion}
$\ML(\triangledown)$ is closed under unions.
\end{Proposition}
\begin{proof}
By Proposition~\ref{mlunion} (b) \ML is closed under unions. The proof 
is done by structural induction on $\varphi$. We add the case of $\varphi$ containing the nonemptiness operator. 

Let $\varphi = \neop \psi $. Assume $K,T_i \models \varphi$ for all $i \in I$. For each $i \in I$ either $T_i$ is the empty set or there exists $S_i \subseteq T_i$ such that $S_i \neq \emptyset$ and $K,S_i \models \psi$. If all $T_i$ are empty $\bigcup_{i \in I} T_i$ is the empty set. Otherwise  $\bigcup_{i \in I} S_i$ is a nonempty subset of $\bigcup_{i \in I} T_i$ statisfying $\psi$. Thus, $K, \bigcup_{i \in I} T_i \models \neop \psi.$ 
\end{proof}

\begin{Proposition} \label{nedisempty}
$\ML(\triangledown)$ has the empty team property.
\end{Proposition}
\begin{proof}
By Proposition~\ref{mlunion} (c) \ML has the empty team property. The proof is done by structural induction on $\varphi$. We add the case of $\varphi$ containing the nonemptiness operator. 

Let $\varphi = \neop \psi $. Then $K,T \models \varphi$ if $T=\emptyset$ by the definition of the nonemptiness operator.
\end{proof}

%

Instead of the inclusion atom as in Lemma~\ref{k-equiv} we now use the nonemptiness operator to extend Lemma \ref{k-equiv1} to cover also the domain totality condition $(\mathrm{D}_k)$ of team $k$-bisimilarity (except for $T'=\emptyset$). 

\begin{Lemma}\label{k-equiv2}
Let $\Phi$ be a finite set of propositon symbols and $k \in \mathbb{N}$. For any pair $(K,T) \in \cK\! \!\cT (\Phi)$ there exists a formula $\zeta_{K,T} \in \ML(\triangledown) (\Phi)$ such that $K',T' \models \zeta_{K,T }$ if and only if $K, T [\rightleftarrows_k] K',T'$ or $T'=\emptyset$.
\end{Lemma}
\begin{proof}
If $T$ is the empty team, we let $\zeta_{K,T } = p\land\lnot p$ for some $p\in\Phi$. Then $K',T' \models \zeta^k_{K,T}$ if and only if $T'=\emptyset$. Else we let $\zeta_{K,T}$ to be the formula 
\[\eta_{K,T}^k\land  \bigwedge_{w \in T} \neop \chi_{K,w}^k,\] 
where $\eta_{K,T}^k$ is as in Lemma~\ref{k-equiv1} and $\chi_{K,w}^k$ is the $k$-th Hintikka-formula of the pair $(K, w)$. 

Assume first that $K',T' \models \zeta_{K,T}$, but $T' \neq \emptyset$. As $K', T'  \models  \bigwedge_{w \in T} \neop \chi_{K,w}^k$ it holds that for all $w \in T$ there exists a nonempty subset $T'_w \subseteq T'$ such that $K', T'_w \models \chi_{K,w}^k$. By  Proposition~\ref{T->v} for all $w' \in T'_w$ it holds that $K', w' \models \chi_{K,w}^k$. It follows that by Proposition \ref{hintikka} for all $w \in T$ there exists $w' \in T'$ such that $K,w \rightleftarrows_k K',w'$. From Lemma~\ref{k-equiv1} it follows that for all $w' \in T'$ there exists a $w \in T$ such that  $K,w \rightleftarrows_k K',w'$. Therefore, $K, T [\rightleftarrows_k] K',T'$. 

For the other direction assume that $K, T [\rightleftarrows_k] K',T'$ or $T'=\emptyset$. If $T'=\emptyset$, then $K',T'\models\zeta^k_{K,T}$ by the empty team property, so we may assume that $K, T [\rightleftarrows_k] K',T'$. Then for every $w'\in T'$ there exists $w\in T$ such that $K, w \rightleftarrows_k K',w'$, and hence by Lemma~\ref{k-equiv1}, $K',T'\models\eta_{K,T}^k$. Additionally for all $w \in T$ there exists $w' \in T'$ such that $K, w \rightleftarrows_k K',w'$. Thus, by Proposition \ref{hintikka} for every $w \in T$ there exists some $w' \in T'$ such that $K',w' \models \chi_{K,w}^k$. Thus, $K', T'$  satisfies $\neop \chi_{K,w}^k$ for every $w \in T$ and therefore $K', T' \models  \bigwedge_{w \in T} \neop \chi_{K,w}^k $. Hence, $K',T' \models \zeta_{K,T}$.
\end{proof}

Next we characterise the expressive power of modal logic with nonemptiness operator in a similar way as in the case of modal logic with intuitionistic disjunction in \cite{epmdl}.

\begin{theorem}\label{disjcharact}
 Let $\Phi$ be a finite set of proposition symbols and let $\cK \subseteq \cK\! \!\cT(\Phi)$. The class $\cK$ is definable in $\ML(\triangledown)$ if and only if $\cK$ closed under unions and closed under $k$-bisimulation for some $k \in \mathbb{N}$ and satisfies the empty team property.
\end{theorem}

\begin{proof}
Substitute $\psi_{K,T}$ in the proof of Theorem~\ref{mincchar} with $\zeta_{K,T}$. Otherwise the proofs are the same. \qedhere
\end{proof}

From this result and Theorem \ref{mincchar} it follows that the expressive power of $\ML(\triangledown)$ and \MINC coincide.

\begin{Corollary} \label{maincor}
$\ML( \triangledown) \equiv \MINC$.
\end{Corollary}

\subsection*{Nonempty disjunction}
Another way to express nonemptiness is to define a disjunction which requires that each of the disjuncts is satisfied in a nonempty subset. This connective was first suggested by Raine Rönnholm \cite{raine}. 

\begin{Definition}
$K, T \models  \varphi \nedis \psi$ if and only if $T= \emptyset$ or there exist $T_1, T_2 \neq \emptyset$ such that $T= T_1 \cup T_2$ and $K, T_1 \models \varphi$ and $K, T_2 \models \psi$.
\end{Definition}

Closure under $k$-bisimulation, closure under unions and the empty team property also hold for modal logic with nonempty disjunction. Using the nonempty disjunction produces results analogous to those obtained in~\cite{epmdl}. The formula $\varphi$ in the proof of Theorem~\ref{disjcharact} can for instance be formulated as
\[\bigvee_{(K,T) \in \cK}  \bnedis\displaylimits_{w\in T} \chi_{K,w}^k.\]

It is easy to see that the nonempty disjunction has the same expressive power as the nonemptiness operator.
\begin{Proposition}\label{nedisjunct} The following equivalences hold:
\begin{enumerate}[i)]
\item $\neop \varphi \equiv \varphi \nedis \top$
\item $\varphi \nedis \psi \equiv (\varphi \lor \psi) \land (\neop \varphi \land \neop \psi)$
\end{enumerate}
\end{Proposition}
%


\section{A lower bound for the translation}

By Corollary \ref{maincor} we know that for each formula in \MINC there exists an equivalent formula in $\ML(\triangledown)$, and vice versa. To establish a lower bound for the size of the translation from \MINC to $\ML(\triangledown)$ we will use the semantic game of $\ML(\triangledown)$ where one player tries to show that a formula holds, while the other player tries to refute the formula. 
\smallskip

The following idea of using a function $F$ as a winning strategy in the semantic game is due to Jonni Virtema and is used in \cite{jonni}. Raine Rönnholm \cite{raine} introduced the idea in the context of inclusion and exclusion logic. 

\begin{Definition}\label{game}
Let $K=(W,R,V)$ be a Kripke model and $T$ a team in $W$.
Let $\varphi \in \ML(\triangledown)$. Let $F$ be a function mapping nodes $\psi$ in the syntax tree of $\varphi$ to subsets of $W$. We say that $F$ is a winning strategy in the semantic game for $(K,T)$ and $\varphi$ if it satisfies the following conditions:
\begin{enumerate}[i)]
\item $F(\varphi)=T$
\item If $\psi=p$ then $F(p)$ is such that $K,F(p) \models p$.
\item If $\psi=\neg p$ then $F(\neg p)$ is such that $K,F(\neg p) \models \neg p$.
\item If $\psi = \theta \land \eta$ then $F(\psi)=F(\theta)=F(\eta)$.
\item If $\psi = \theta \lor \eta$ then $F(\psi)=F(\theta)\cup F(\eta)$.
\item If $\psi = \neop \theta$ then $F(\psi)\supseteq F(\theta)$ and if $F(\psi) \neq \emptyset$ then $F(\theta) \neq \emptyset$.
\item If $\psi= \Diamond \theta$ then $F(\psi)[R]F(\theta)$
\item If $\psi= \Box \theta$ then $F(\theta)= R[F(\psi)]$
\end{enumerate}
\end{Definition}

In the semantic game such a function $F$ is a winning strategy for the first player, who tries to show that the formula holds. Thus, the following lemma holds. 

\begin{Lemma}\label{winningstrategy}
A function $F$ as given above exists if and only if $K, T \models \varphi$.
\end{Lemma}
\begin{proof}
Assume that $K,T \models \varphi$.  Using traversal of the syntax tree of $\varphi$ from the root we will define $F$ such that if a node $\psi$ satisfies $K, F(\psi) \models \psi$  then each child $\theta$ of this node satisfies $K, F(\theta) \models \theta$. 

 \begin{enumerate}[i)]
\item Define $F(\varphi)=T$. Then by assumption $K, F(\varphi) \models \varphi$.
\setcounter{enumi}{3}
\item If $\psi = \theta \land \eta$, then by assumption $K, F(\psi) \models \psi$, thus $K, F(\psi) \models \theta$ and $K, F(\psi) \models \eta$ . We define $F(\theta)=F(\eta)= F(\psi)$. Then $K, F(\theta) \models \theta$ and $K, F(\eta) \models \eta$. 
\item If $\psi = \theta \lor \eta$, then by assumption $K,F(\psi) \models \psi$ and thus there exist subsets $T_1, T_2$ of $F(\psi)$ such that $F(\psi)= T_1 \cup T_2$ and $K,T_1 \models \theta$ and $K, T_ 2\models \eta$. We define  $F(\theta)= T_1$ and $F(\eta)=T_2$. It follows that $F(\psi)=F(\theta)\cup F(\eta)$ and $K, F(\theta) \models \theta$ and $K, F(\eta) \models \eta$.
\item If $\psi = \neop \theta $, then by assumption $K,F(\psi) \models \psi$ and thus there exist a subset $T' \subseteq F(\psi)$, with $T' \neq \emptyset$ if $F(\psi)\neq \emptyset$, such that  $K,T'\models \theta$. We define $F(\theta)=T'$. It follows that $F(\psi) \supseteq F(\theta)$,  $K, F(\theta) \models \theta$ and if $F(\psi) \neq \emptyset$ then $F(\theta) \neq \emptyset$.
\item If $\psi= \Diamond \theta$, then by assumption $K,F(\psi) \models \psi$ and thus there exists a set $T'$ such that $K,T' \models \theta$ and $F(\psi)[R]T'$. We define $F(\theta)= T'$. Thus $F(\psi)[R]F(\theta)$ and $K, F(\theta) \models \theta$.
\item If $\psi= \Box \theta$, then by assumption $K,F(\psi) \models \psi$ and thus $K,T' \models \theta$ for $T'=R[F(\psi)]$. We define $F(\theta)= R[F(\psi)]$. Thus $K, F(\theta) \models \theta$.
 \end{enumerate}
Clearly if $\psi=p$ or $\psi= \neg p$ then by assumption $K,F(p) \models p$ or $K,F(p) \models \neg p$. Thus $F$ is satisfies the conditions in Definition \ref{game}.

For the other direction assume that $F$ exists as given in Definition \ref{game}. We will show by induction on the nodes in the syntax tree of $\varphi$ that for each node $\psi$ it holds that $K, F(\psi) \models \psi$. Since $F(\varphi)=T$ it follows that $K, T \models \varphi$.
\begin{enumerate}[i)]
\setcounter{enumi}{1}
\item If $\psi = p$ then by assumption $F(p)$ is such that $K,F(p) \models p$.
\item If $\psi = \neg p$ then by assumption $F(\neg p)$ is such that $K,F(\neg p) \models \neg  p$.
\item If $\psi = \theta \land \eta$ then $F(\psi)=F(\theta)=F(\eta)$. By induction hypothesis $K, F(\theta) \models \theta$ and $K, F(\eta) \models \eta$. Thus, $K, F(\psi) \models \psi$.
\item If $\psi = \theta \lor \eta$ then $F(\psi)=F(\theta)\cup F(\eta)$. By induction hypothesis  $K, F(\theta) \models \theta$ and $K, F(\eta) \models \eta$. Thus, $K, F(\psi) \models \psi$.
\item If $\psi = \neop \theta$ then $F(\psi) \supseteq F(\theta)$ and if $F(\psi) \neq \emptyset$ then $F(\theta) \neq \emptyset$. By induction hypothesis  $K, F(\theta) \models \theta$. Thus, $K, F(\psi) \models \psi$.
\item If $\psi= \Diamond \theta$ then $F(\theta)[R]F(\psi)$. By induction hypothesis $K, F(\theta) \models \theta$. Thus, $K, F(\psi) \models \psi$.
\item If $\psi= \Box \theta$ then $F(\theta)= R[F(\psi)]$. By induction hypothesis $K, F(\theta) \models \theta$. Thus, $K, F(\psi) \models \psi$.
\end{enumerate}
As $F(\varphi)=T$ it follows finally that $K,T \models \varphi$.
\end{proof}

Let $K=(W,R,V)$ be a fixed Kripke model, where $R$ is the identity relation. In this case elements are only in relation to themselves. Thus, $\Diamond \psi \equiv \psi$ and $\Box \psi \equiv \psi$ for all $\psi \in \ML(\triangledown)$ on this model $K$.  Let $\varphi$ be a formula in $\ML(\triangledown)$ and let $T$ be a sufficiently large team in $K$ such that $K,T \models \varphi$. Let $F$ be a function as given in Definition \ref{game}. Now remove an element $a$ from $T$. Define $F'(\psi)= F(\psi) \backslash \{a\}$ for each $\psi$ in the syntax tree of $\varphi$. Assume that $F(\theta) \neq \{a\}$  for all subformulas of $\varphi$ which are of the form $\neop \theta$. We will show by induction on $\varphi$ that $F'$ satisfies the conditions given in Definition \ref{game} for $T \setminus \{a\}$ in place of $T$. 

 \begin{enumerate}[i)]
\item $F'(\varphi)=T \backslash \{a\}$.
\item $F'(p) \subseteq F(p) \subseteq V(p)$, thus $K, F'(p) \models p$. 
\item $F'(\neg p) \subseteq F(\neg p) \subseteq W\setminus V(p)$, thus $K, F'(\neg p) \models \neg p$. 
\item If $\psi = \theta \land \eta$, then by definition $F'(\psi) = F(\psi) \setminus \{a\}= F(\theta) \setminus \{a\} = F(\eta) \setminus \{a\}$, thus $F'(\psi)=F'(\theta)=F'(\eta)$.
\item If $\psi = \theta \lor \eta$, then by definition $F'(\psi) = F(\psi) \setminus \{a\}= F(\theta) \setminus \{a\}  \cup F(\eta) \setminus \{a\}=F'(\theta) \cup F'(\eta)$.
\item If $\psi = \neop \theta$, then by definition $F'(\psi) = F(\psi) \setminus \{a\} \supseteq F(\theta) \setminus \{a\}= F'(\theta)$ and $F'(\theta) \neq \emptyset$, since $F(\theta) \neq \{a\}$.
\item If $\psi= \Diamond \theta$, then as $\psi \equiv \theta$ we see that $F(\psi)=F(\theta)$. Then, $F'(\psi)=F'(\theta)$ and it follows that $F'(\psi)[R]F'(\theta)$.
\item If $\psi= \Box \theta$, then  as $\psi \equiv \theta$ we see that $F(\psi)=R[F(\psi)]=F(\theta)$. Thus, $F'(\psi)=R[F'(\psi)]=F'(\theta)$.
\end{enumerate}

Thus, if $F(\theta)\neq \{a\}$ for all subformulas of $\varphi$ which are of the form $\neop \theta$ then $F'$ satisfies the conditions given by Definition \ref{game}, and by Lemma \ref{winningstrategy} it holds that $K, T \backslash \{a\} \models \varphi$. This means that there are at most $ \occv(\varphi)$ elements which cannot be removed from the team without changing the truth of $\varphi$. Here $\occv(\varphi)$ is the number of occurrences of the operator $\triangledown$  in the formula $\varphi$.

Assume $K,T \models p_1\ldots p_n \subseteq q_1\ldots q_n$. Let $\varphi \in \ML(\triangledown)$ be such that it defines this inclusion atom.  We will look at the set $A_T=\{a \in T \mid K, T \backslash \{a\} \nmodels p_1\ldots p_n \subseteq q_1\ldots q_n \}$ of those elements of $T$ which are essential for the satisfaction of the inclusion atom. If $|A_T| > \occv(\varphi)$, then at least one element $a$ which does not interfere with the satisfaction of the formula $\varphi$ can be removed: $K, T\backslash \{a\}\models \varphi$ but $K,T \backslash \{a\} \nmodels p_1\ldots p_n \subseteq q_1\ldots q_n$. Therefore, $\varphi$ does not define the inclusion atom $ p_1\ldots p_n \subseteq q_1\ldots q_n$. It follows that a formula  $\varphi \in \ML(\triangledown)$ which defines the inclusion atoms needs to satisfy  $\occv(\varphi) \geq |A_T|$.

\smallskip

\hspace{-7mm}
\begin{minipage}[c]{0.8 \textwidth}
\begin{example} \label{example}
Let $\Phi=\{p_1, p_2, q_1, q_2\}$ and $K=(W,R,V)$ a Kripke model, where $W=\{w_{a_1 a_2 b_1 b_2} \mid a_1,a_2,b_1,b_2 \in \{0,1\}\}$, $R= \text{id}$, $V(p_i)=\{w_{a_1 a_2 b_1 b_2} \mid a_i =1\}$ and $V(q_j)=\{w_{a_1 a_2 b_1 b_2} \mid b_j =1\}$. The team $T=\{w_{0001},w_{0110}, w_{1011}, w_{1100} \}$ satisfies the inclusion atom $p_1 p_2 \subseteq q_1 q_2$, but no proper subteam of $T$ satisfies it. Thus, there are at least four occurrences of $\triangledown$ needed to describe this inclusion atom in $\ML(\triangledown)$.

\end{example}
The way the team $T$ in the example is built can be extended in a straightforward manner to arbitrary arity $n$ of the inclusion atom. 
\end{minipage}
\begin{minipage}[c]{0.25\textwidth}
\[
\xy
\xygraph{%
!{<0cm,0cm>;<1cm,0cm>:<0cm,1cm>::}
!{(0,2)}*{p_1 p_2}="p"
!{(1,2)}*{ q_1 q_2}="q"
!{(0,0)}*{\bullet_{00}}="00"
!{(0,0.5)}*{\bullet_{01}}="01"
!{(0,1)}*{\bullet_{10}}="10"
!{(0,1.5)}*{\bullet_{11}}="11"
!{(1,0)}*{\bullet_{00}}="200"
!{(1,0.5)}*{\bullet_{01}}="201"
!{(1,1)}*{\bullet_{10}}="210"
!{(1,1.5)}*{\bullet_{11}}="211"
"00"-"201" "01"-"210" "10"-"211" "11"-"200"
}
\endxy
\]
\end{minipage}

\smallskip

\begin{theorem} \label{size}
A formula $\varphi$  in modal logic with  nonemptiness operator which describes the inclusion atom $p_1\ldots p_n \subseteq q_1\ldots q_n$ contains at least $ 2^{n}$ symbols.
\end{theorem}
\begin{proof}
Let $K=(W,R,V)$ be a Kripke model, where $R$ is the identity relation. For each $n$-ary inclusion atom $p_1\ldots p_n \subseteq q_1\ldots q_n$ there exists a team $T \subseteq \{w_{\bar{a}\bar{b} }\mid \bar{a}, \bar{b} \in \{0,1\}^n\}$ in $\|\varphi\|^K$, in which the tuples $\bar{a}$ and $\bar{b}$ form a cyclic permutation of $n$-tuples on the set $\{0,1\}$, as in Example~\ref{example}. Therefore, $T$ is of size $2^n$ and no proper subteam satisfies the inclusion atom. Thus, $A_T=T$. As  $\occv(\varphi) \geq  |A_T|$ if follows that $\occv(\varphi) \geq 2^{n}$.
\end{proof}


\begin{Remark}
The case of proposition symbols $p_i$ and $q_i$, $1 \leq i \leq n$, in the inclusion atoms is the most general case in the sense that arbitrary formulas $\varphi_i$ and $\psi_i$ can only reduce the options for $T$, as those depend on how the formulas are evaluated.
\end{Remark}

\section{Concluding remarks}

In this paper we characterised the expressive power of modal inclusion logic by closure under $k$-bi\-simulation, closure under unions and the empty team property. We also showed that the same characterisation holds for modal logic with nonemptiness operator, and therefore it has the same expressive power. Furthermore, we established a lower bound for the size of the translation from  $\MINC$ to $\ML( \triangledown)$: to describe the inclusion atom $p_1\ldots p_n \subseteq q_1\ldots q_n$ in terms of modal logic with nonemptiness operator a formula of at least $2^{n}$ occurrences of the operator $\triangledown$ is required. 

Note that in the characterisation of the expressive power of $\MINC$ only unary inclusion atoms of the form $\chi_{K,u}^k  \subseteq \chi_{K,v}^k$ were used. Thus, each $n$-ary inclusion atom of the form $\varphi_1 \ldots \varphi_n \subseteq \psi_1\ldots  \psi_n$ can be expressed in terms of unary inclusion atoms $\varphi \subseteq  \psi$. 
%
%
In \cite{HKMV15} it was shown that the computational complexity of the satisfiability problem of modal inclusion logic is complete for EXPTIME\footnote{This holds for \MINC with lax semantics. With strict semantics, the satisfiability problem of \MINC is complete for NEXPTIME, see \cite{HKMV15}.}. Using inclusion atoms of arbitrary arity is crucial for the proof of this result. This raises the question: what is the computational complexity of the satisfiability problem of modal logic with only unary inclusion atoms?

\bibliographystyle{eptcs}
\bibliography{eminc}

\end{document}